\documentclass[pdflatex,sn-mathphys-num]{sn-jnl}

\usepackage{float} 
\usepackage{graphicx}
\usepackage{multirow}%
\usepackage{amsmath,amssymb,amsfonts}%
\usepackage{amsthm}%
\usepackage{caption}
\usepackage{mathrsfs}%
\usepackage[title]{appendix}%
\usepackage{xcolor}%
\usepackage{textcomp}%
\usepackage{manyfoot}%
\usepackage{booktabs}%
\usepackage{algorithm}%
\usepackage{algorithmicx}%
\usepackage{algpseudocode}%
\usepackage{listings}%
\usepackage{tikz}%
\usetikzlibrary{arrows.meta}%

\theoremstyle{thmstyleone}%
\newtheorem{theorem}{Theorem}
\newtheorem{corollary}{Corollary}

\theoremstyle{thmstyletwo}%
\newtheorem{remark}{Remark}%

\theoremstyle{thmstylethree}%
\newtheorem{definition}{Definition}%

\begin{document}

\title[An Epidemic Threshold Set for Networks]{An Epidemic Threshold Set for Networks}


\author*[1]{\fnm{Hoang Phi Dung}}\email{dunghp@ptit.edu.vn}

\author[2]{\fnm{Nguyen Hong Phuc}}\email{hongphuc692k7@gmail.com}
\equalcont{These authors contributed equally to this work.}

\equalcont{These authors contributed equally to this work.}

\affil*[1]{\orgdiv{Department of Mathematics - Faculty of Fundamental Sciences}, \orgname{Posts and Telecommunications Institute of Technology}, \orgaddress{\street{96A Tran Phu}, \city{Ha Dong}, \postcode{12100}, \state{Hanoi}, \country{Vietnam}}}

\affil[2]{\orgdiv{Faculty of Information Technology}, \orgname{Posts and Telecommunications Institute of Technology}, \orgaddress{\street{96A Tran Phu}, \city{Ha Dong}, \postcode{12100}, \state{Hanoi}, \country{Vietnam}}}



\abstract{
In this paper, we investigate a discrete-time SIS epidemic model and the epidemic thresholds on complex networks. We focus on proposing a community-level epidemic threshold set and establishing a comparative result between the local epidemic thresholds and the global epidemic threshold. To verify our theoretical findings and structural properties, we conduct numerical experiments on one synthetic network (Network1) and one real-world network (the Haslemere contact network). Our numerical simulations, along with the computation and statistical ranking of the epidemic threshold sets, align accurately with our theoretical results.
}

\maketitle
\section{Introduction}
In the modern era, complex networks—such as online social networks, computer systems, and the Internet of Things (IoT)—have transformed global communication. However, this high connectivity also accelerates the rapid diffusion of unwanted spreading processes, including misinformation, computer viruses, and epidemic pathogens \cite{Pham2019, Acemoglu2013, Chen2014, Cho2025}. Following foundational breakthroughs in complex network theory \cite{Watts1998, Barabasi1999}, modeling and analyzing these propagation dynamics has emerged as a vital research area with profound implications for information security, public health, and social sciences \cite{Pastor-Satorras2015, Nowzari2016, Chen2014}.

Epidemic modeling on networks is largely rooted in classical biological frameworks, such as the Susceptible-Infected-Recovered (SIR) and Susceptible-Infected-Susceptible (SIS) models \cite{Ahn2013, Chen2015, Pare2020, Khanh2026}. Among these, the discrete-time SIS model on networks has been widely studied. A fundamental question in this domain is determining the conditions under which an epidemic persists or dies out—a transition point mathematically defined as the global epidemic threshold \cite{Wang2003, Chakrabarti2008, Mieghem2009}.

However, real-world networks are rarely homogeneous. Instead, they typically exhibit distinct community structures—dense clusters of internal connections paired with sparse links between different communities \cite{Girvan2002, Newman2006, Fortunato2010}. This structural heterogeneity causes local spreading dynamics to diverge significantly across the network \cite{Stegehuis2016, Salathe2010}.

In this paper, we consider a discrete-time SIS spreading model analogous to those investigated in \cite{Wang2003, Chakrabarti2008}, which exhibits distinct structural differences from the SIS variations discussed in \cite{Ahn2013, Gracy2021, Pare2020, Dung2026}. 

\vspace{1em}
\noindent\textit{Contributions}

A major limitation of existing epidemic threshold studies \cite{Wang2003, Chakrabarti2008, Mieghem2009, Mieghem2012, Cantwell2026, Ruhi2016} is their neglect of community structures, as a single global epidemic threshold $\tau = \frac{1}{\lambda_1(A)}$ cannot capture local vulnerabilities. To address this, we introduce a community-level epidemic threshold set \eqref{thresholdprofile} to quantify and compare local epidemic thresholds among all individual communities. Mathematically, we prove Theorem \ref{1vsall}, establishing that any local epidemic threshold is bounded below by the global epidemic threshold. We then validate our theoretical findings and Corollaries \ref{hequa1}--\ref{hequa2} through numerical simulations. Ultimately, this localization provides a community-oriented perspective for analyzing propagation. Future research could integrate this approach with tighter threshold bounds—such as those proposed by Van Mieghem et al. \cite{Mieghem2012} or Ruhi et al. \cite{Ruhi2016}—to further improve targeted mitigation strategies.

\vspace{1em}
\noindent\textit{Outline}

The remainder of this paper is organized as follows. Section \ref{sec2} describes the investigated epidemic model and formalizes the definition of the epidemic threshold. In Section \ref{sec3}, we propose the community-level epidemic threshold set and establish a comparative theorem regarding the relationship between local epidemic thresholds and the global epidemic threshold. Subsequently, in Section \ref{sec4}, we compute the epidemic threshold sets and measure the aggregate infection quantity of each community via numerical simulations on real and synthetic datasets. Finally, we conclude the paper and discuss future research directions in Section \ref{conclude}.

\section{Preliminaries} \label{sec2}

\subsection{The Discrete-Time SIS Model}
Consider an undirected, connected network $G=(V,E)$ with $|V|=n$. Let $A=(A_{ij})_{n\times n}$ be the adjacency matrix of $G$, where:

$$
A_{ij}=
\begin{cases}
1, & \text{if there is an edge between node } i \text{ and node } j,\\
0, & \text{otherwise.}
\end{cases}
$$

In the SIS model, each node exists in exactly one of two discrete states at any given time step: Susceptible ($S$) or Infected ($I$). During a single time step, the disease transmission probability across any edge is $\beta$; that is, a susceptible node $S$ can be infected by each of its infectious neighbors $I$ with probability $\beta$. Simultaneously, an infected node $I$ recovers and returns to the susceptible state $S$ with probability $\delta$ per time step. Note that we assume $\beta$ and $\delta$ are homogeneous across the entire network $G$. Henceforth, we uniformly refer to $\beta$ and $\delta$ as the transmission rate and recovery rate, respectively. This epidemic propagation process can be rigorously modeled as a Markov chain with $2^n$ states. Wang et al. \cite{Wang2003} and Chakrabarti et al. \cite{Chakrabarti2008} investigated a discrete-time SIS Model detailed below.

Let $p_i(t)$ denote the probability that node $i$ is infected at time step $t$. Let $N(i)$ be the set of neighbors of node $i$, i.e., $N(i)=\{j\in V:A_{ij}=1\}$. The probability that node $i$ is not infected by any of its neighbors at time $t$ is a random variable given by:
\begin{align}\label{neighborhood-infected}
\zeta_i(t)=\prod\limits_{j\in N(i)}
\left(1-\beta p_j(t-1)\right) =\prod_{j=1}^n
\left(1-\beta A_{ij}p_j(t-1)\right).
\end{align}
We assume that probabilities $p_j(t-1)$ are
independent of each other.

Using Equation \eqref{neighborhood-infected}, we obtain the discrete-time SIS propagation model that we investigate throughout this paper (\cite[Equation (7)]{Chakrabarti2008}):
\begin{align}\label{main-model} 
    p_i(t+1) = 1 - \left[(1-p_i(t)) + \delta p_i(t)\right]\prod_{j=1}^n
\left(1-\beta A_{ij}p_j(t)\right).
\end{align}
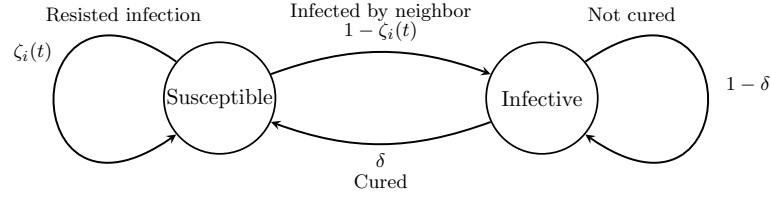
\begin{figure}[h]
\centering
\resizebox{0.78\linewidth}{!}{%
\begin{tikzpicture}[
    state/.style={
        circle,
        draw,
        line width=0.8pt,
        minimum size=1.8cm,
        inner sep=0pt,
        font=\normalsize,
        align=center
    },
    trans/.style={
        ->,
        >=stealth,
        line width=0.9pt
    },
    lab/.style={
        font=\small
    }
]

\node[state] (S) at (0,0) {Susceptible};
\node[state] (I) at (5.2,0) {Infective};

\draw[trans]
    (S.25) to[out=20,in=160]
    node[midway, above=0.38cm, lab] {Infected by neighbor}
    node[midway, above=0.05cm, lab] {$1-\zeta_i(t)$}
    (I.155);

\draw[trans]
    (I.205) to[out=200,in=-20]
    node[midway, below=0.02cm, lab] {$\delta$}
    node[midway, below=0.35cm, lab] {Cured}
    (S.335);

\draw[trans]
    (S.140) to[out=145,in=215,looseness=7]
    node[pos=0.32, left=0.15cm, lab] {$\zeta_i(t)$}
    (S.220);

\node[lab] at (-1.55,1.35) {Resisted infection};

\draw[trans]
    (I.40) to[out=35,in=-35,looseness=7]
    node[pos=0.45, right=0.18cm, lab] {$1-\delta$}
    (I.320);
\node[lab] at (6.65,1.35) {Not cured};
\end{tikzpicture}%
}
    \caption{State transition diagram of the SIS model in \cite{Chakrabarti2008}.}
\label{fig:sis_transition_diagram}
\end{figure}

\begin{remark} 
\begin{enumerate}
    \item[(i)] The discrete-time SIS model in \cite[Equation (7)]{Chakrabarti2008} slightly modifies the earlier formulation in \cite[Equation (6)]{Wang2003}. In \cite{Chakrabarti2008}, the authors neglected the higher-order term $\dfrac{1}{2}\delta p_i(t)\left(1 - \zeta_i(t+1)\right)$. This neglect is justified because the model in \cite{Wang2003} can be asymptotically approximated to yield Equation (7) in \cite{Chakrabarti2008}, which corresponds precisely to Model \eqref{main-model} investigated herein.
    \item[(ii)] In the investigated SIS model, the recurrence relation is non-linear because the probability that a node resists infection from its neighbors is expressed as a product. Consequently, this system is governed by a \textit{Nonlinear Dynamical System} (NLDS). 
\end{enumerate}
\end{remark}
\subsection{Epidemic Threshold of a Propagation Model with Homogeneous Transmission and Recovery Rates}
We recall the formal definition of the epidemic threshold on a network as formulated by Chakrabarti et al. \cite{Chakrabarti2008} (see also \cite{Wang2003}).
\begin{definition}[Epidemic Threshold]\label{threshold} A positive real value $\tau$ is defined as the epidemic threshold for NLDS on a network if it satisfies the following two conditions:
\begin{enumerate}
    \item[(i)] If $\dfrac{\beta}{\delta}<\tau$, the epidemic dies out exponentially over time.
    \item[(ii)] If $\dfrac{\beta}{\delta}>\tau$, the epidemic persists and can break out across the network.
\end{enumerate}
Here, the ratio $\frac{\beta}{\delta}$ is termed the \textit{effective spreading rate} (\cite[Equation (5.1)]{Chen2015}).
\end{definition}
Thus, when a network's effective spreading rate exceeds its epidemic threshold, a disease outbreak occurs. Conversely, when the effective spreading rate falls below the threshold, the epidemic dies out over time.

In \cite{Wang2003} and \cite{Chakrabarti2008}, the authors defined an explicit epidemic threshold directly determined by the network's adjacency matrix $A$.

\begin{definition}[\cite{Chakrabarti2008}]\label{eigen}
The value $\tau=\frac{1}{\lambda_1(A)}$ is defined as the epidemic threshold of the network $G$, where $\lambda_1(A)$ denotes the largest eigenvalue (spectral radius) of the adjacency matrix $A$. 
\end{definition}
Following the calculations in \cite[Appendix A]{Chakrabarti2008}, the condition for the disease-free state, where $p_i(t) \to 0$ for all $i \in V$ as $t \to \infty$, is $\frac{\beta}{\delta} < \frac{1}{\lambda_1(A)}$.

\section{Community-Level Epidemic Threshold Set}\label{sec3}
\subsection{Epidemic Threshold Set with Homogeneous Transmission and Recovery Rates}
Given a network with $K$ disjoint, connected communities, we formulate the SIS model \eqref{main-model} on the induced subgraph of each community $C_g$ independently. This setup temporarily disregards inter-community edges while keeping the transmission ($\beta$) and recovery ($\delta$) rates unchanged.

Let $A_g$ be the adjacency matrix of $C_g$. Following Definition \ref{eigen}, the local epidemic threshold of community $C_g$ is defined as:
$$\tau_g = \frac{1}{\lambda_1(A_g)}, \quad g = 1, \dots, K$$
Hereafter, $\tau$ and $\tau_g$ denote the global and local epidemic thresholds, respectively.

\begin{definition}[Epidemic Threshold Set]\label{thresholdprofile}
The ordered set of positive real numbers
\[\mathcal{T}_A = (\tau_1,\tau_2,\ldots,\tau_{K})\]
is defined as the \textit{epidemic threshold set} of the network $G$, where $\tau_g =\frac{1}{\lambda_1(A_g)}$ denotes the local epidemic threshold associated with community $C_g$, for $g = 1, \dots, K$.
\end{definition}
We can rank the communities in ascending order of their local epidemic thresholds (re-indexing if necessary) such that:
\[\tau_{1}\leq \tau_{2}\leq \cdots \leq \tau_{{K}}.\]

\subsection{Comparison and Evaluation}
\vspace{0.2cm}

After establishing the community-based threshold set, we compare the local epidemic threshold of an individual community with the global epidemic threshold.

By Definition \ref{eigen}, the global epidemic threshold is given by $\tau=\frac{1}{\lambda_1(A)}$. The following theorem establishes a fundamental inequality connecting the global epidemic threshold to the local epidemic thresholds.

\begin{theorem}\label{1vsall}
The local epidemic threshold of any community is greater than or equal to the global epidemic threshold. That is, for every community $C_g$, we have:
\[\tau_g\geq \tau.\]
Furthermore, it holds that $\tau \le \tau_1 \le \dots \le \tau_K$.
\end{theorem}

\begin{proof}
Let $P$ be a permutation matrix that groups all nodes belonging to community $C_g$ into a single contiguous block within the matrix $A$. We can express the transformed matrix as:
\[\hat{A}=
P^T A P
=
\begin{pmatrix}
A_g & R\\
R^T & A_{\bar g}
\end{pmatrix},
\]
where $A_g$ is the adjacency matrix of the subgraph induced by community $C_g$, $A_{\bar g}$ is the adjacency matrix of the remainder of the network, and $R$ represents the bipartite cut (inter-community edges) connecting $C_g$ to $V\setminus C_g$.

Since $P$ is an orthogonal permutation matrix ($P^T=P^{-1}$), $A$ and $\hat{A}$ are similar matrices. Consequently, they share the exact same set of eigenvalues (spectrum). 

Let $u_g\in \mathbb{R}^{|C_g|}$ be the eigenvector corresponding to the largest eigenvalue $\lambda_1(A_g)$ of $A_g$, meaning $A_g u_g = \lambda_1(A_g)u_g$. We construct an expanded vector in $\mathbb{R}^n$ by zero-padding the remaining entries:
$
\tilde u_g
=
\begin{pmatrix}
u_g\\
0
\end{pmatrix}.
$
According to the Rayleigh quotient theorem for real symmetric matrices (see \cite{Horn2012}), the largest eigenvalue of $\hat{A}$ satisfies:
\[
\lambda_1(\hat{A})
=
\max_{x\neq 0}
\frac{x^T \hat{A} x}{x^T x}.\]

Since $\tilde u_g\neq 0$, it follows directly that $
\lambda_1(\hat{A}) \geq
\frac{\tilde u_g^T \hat{A}\tilde u_g}{\tilde u_g^T\tilde u_g}.$ Then, we have:
\[
\tilde u_g^T \hat{A}\tilde u_g
=
\begin{pmatrix}
u_g^T & 0
\end{pmatrix}
\begin{pmatrix}
A_g & R\\
R^T & A_{\bar g}
\end{pmatrix}
\begin{pmatrix}
u_g\\
0
\end{pmatrix}
=
u_g^T A_g u_g.
\]

Simultaneously, we have $
\tilde u_g^T\tilde u_g = u_g^T u_g.$
Therefore, we obtain:
\[
\frac{\tilde u_g^T \hat{A}\tilde u_g}{\tilde u_g^T\tilde u_g}
=
\frac{u_g^T A_g u_g}{u_g^T u_g}
=
\lambda_1(A_g).
\]

It follows that
$\lambda_1(\hat{A})\geq\lambda_1(A_g)$. Since $A$ and $\hat{A}$ have the same spectrum,
$\lambda_1(A)=\lambda_1(\hat{A})$, and hence
$\lambda_1(A)\geq\lambda_1(A_g)$. Because both eigenvalues are positive,
$\frac{1}{\lambda_1(A_g)}\geq\frac{1}{\lambda_1(A)}$.
Therefore, $\tau_g\geq\tau$, which completes the proof.
\end{proof}
\begin{corollary} \label{hequa1}
If an epidemic cannot break out across the entire network, it cannot break out independently within any individual community.
\end{corollary} 

\begin{proof}
By Theorem \ref{1vsall}, for any community $C_g$, we have $\tau_g\geq \tau$. If the effective spreading rate is below the global epidemic threshold, then $\frac{\beta}{\delta}<\tau$. Consequently, for every community $C_g$, the inequality $\frac{\beta}{\delta}<\tau\leq \tau_g$ holds. Thus $
\frac{\beta}{\delta}<\tau_g, \forall g=1,2,\ldots, K.$
According to Definition \ref{threshold}, since the effective spreading rate is below the local epidemic threshold of every community, the epidemic cannot break out in any community.
\end{proof}

\begin{corollary} \label{hequa2}
If an epidemic is able to break out within any community, it will inevitably trigger an outbreak across the entire global network.
\end{corollary}
\begin{proof}
Suppose there exists a community $C_g$ whose effective spreading rate exceeds its local epidemic threshold, i.e., $\frac{\beta}{\delta}>\tau_g$. By Theorem \ref{1vsall}, we know that $\tau_g\geq\tau$. It immediately follows that $\frac{\beta}{\delta}>\tau_g\geq\tau$, yielding $\frac{\beta}{\delta}>\tau$. According to Definition \ref{threshold}, since the effective spreading rate exceeds the global epidemic threshold, a global epidemic outbreak is guaranteed.
\end{proof}
\section{Numerical Experiments}\label{sec4}
\subsection{Datasets}

We evaluate the proposed epidemic threshold sets on two datasets: one synthetic network (\textbf{Network1}) and one real-world contact network (\textbf{Haslemere}).

For the first dataset, we randomly select the number of communities $K$ uniformly in the range $[8, 12]$, and set the community sizes $|C_g|$ between $45$ and $90$ nodes. For each community, the internal connection probability $p_{\mathrm{in}}$ is chosen randomly from $[0.08, 0.16]$, while the inter-community connection probability $p_{\mathrm{out}}$ between different modules is drawn from $[0.002, 0.012]$. The resulting probability matrix is symmetrized to generate an undirected graph. We designate this synthetic network as \textbf{Network1}.

The second dataset is the \textbf{Haslemere} contact network \cite{Firth2020, Klepac2018}. Here, nodes represent residents of Haslemere, and edges represent interpersonal physical contacts captured via wearable sensors. The raw graph contains $469$ nodes and $1,262$ edges, with a maximum degree of $37$. Because the raw data contains severe high-degree outliers, a preprocessed version was generated by pruning specific edges attached to high-degree hubs, removing isolated nodes, and ensuring that both the overall network and each individual community remain connected. The preprocessed network consists of $369$ nodes and $808$ edges, with a capped maximum degree of $8$.

\subsection{Computation of Community Epidemic Threshold Sets and Comparison with the Global Epidemic Threshold}
In our first experiment, we construct the community-level epidemic threshold set for each network. We evaluate each community $C_g$ and its corresponding induced sub-adjacency matrix $A_g$ independently. The global epidemic threshold of the entire graph is computed via $\tau=\frac{1}{\lambda_1(A)}$, where $A$ is the full adjacency matrix of network $G$. 

For each community $C_g$, we compute and report the following metrics:
\begin{itemize} 
\item $|C_g|$: The number of nodes in community $C_g$. 
\item $|E_g|$: The number of internal edges in community $C_g$ (edges where both endpoints reside within $C_g$). 
\item $\tau_g$: The local epidemic threshold of community $C_g$, calculated as $\tau_g=\frac{1}{\lambda_1(A_g)}$.
\item $\tau_g-\tau$: The spectral gap between the local epidemic threshold of community $C_g$ and the global epidemic threshold.
\item $\mathrm{Rank}$: The ordinal rank of the community sorted in ascending order of its local epidemic threshold $\tau_g$.
\end{itemize}

The statistical results and calculated metrics of two datasets are summarized in the tables below.
\newpage
\begin{table}[t!]
\centering
\scriptsize
\renewcommand{\arraystretch}{1.15}

\begin{minipage}[t]{0.49\textwidth}
\centering
\resizebox{\linewidth}{!}{%
\begin{tabular}{|c|c|c|c|c|c|}
\hline
$\mathrm{Rank}$ & $C_g$ & $|C_g|$ & $|E_g|$ & $\tau_g$ & $\tau_g-\tau$ \\
\hline
1  & $C_{0}$ & 64 & 321 & 0.0935 & 0.0155 \\
2  & $C_{1}$ & 74 & 324 & 0.0991 & 0.0211 \\
3  & $C_{7}$ & 88 & 343 & 0.1108 & 0.0327 \\
4  & $C_{5}$ & 54 & 181 & 0.1198 & 0.0418 \\
5  & $C_{6}$ & 58 & 199 & 0.1203 & 0.0423 \\
6  & $C_{3}$ & 81 & 277 & 0.1244 & 0.0463 \\
7  & $C_{4}$ & 61 & 217 & 0.1283 & 0.0502 \\
8  & $C_{8}$ & 57 & 167 & 0.1400 & 0.0619 \\
9  & $C_{9}$ & 85 & 248 & 0.1432 & 0.0652 \\
10 & $C_{2}$ & 49 & 138 & 0.1470 & 0.0689 \\
\hline
\end{tabular}%
}
\vspace{0.25cm} 
\textbf{Epidemic Threshold Set of Network1}
\end{minipage}
\hfill
\begin{minipage}[t]{0.49\textwidth}
\centering
\resizebox{\linewidth}{!}{%
\begin{tabular}{|c|c|c|c|c|c|}
\hline
$\mathrm{Rank}$ & $C_g$ & $|C_g|$ & $|E_g|$ & $\tau_g$ & $\tau_g-\tau$ \\
\hline
1 & $C_{7}$ & 43 & 72 & 0.2008 & 0.0423   \\
2 & $C_{2}$ & 73 & 125 & 0.2014 & 0.0429   \\
3 & $C_{3}$ & 99 & 179 & 0.2101 & 0.0516   \\
4 & $C_{4}$ & 62 & 97 & 0.2167 & 0.0582  \\
5 & $C_{5}$ & 17 & 22 & 0.2621 & 0.1036  \\
6 & $C_{0}$ & 20 & 23 & 0.2814 & 0.1228  \\
7 & $C_{1}$ & 36 & 39 & 0.3149 & 0.1564   \\
8 & $C_{6}$ & 19 & 20 & 0.3908 & 0.2322   \\
 &  &   &    &  &  \\
 &  &   &    &  & \\
\hline
\end{tabular}%
}
\vspace{0.25cm}
\textbf{Epidemic Threshold Set of Haslemere}
\end{minipage}
\caption{Epidemic Threshold Set of Networks}
\end{table}

Consistent with Theorem \ref{1vsall}, all local epidemic thresholds $\tau_g$ are greater than the global epidemic threshold $\tau$. While these local thresholds differ minimally in synthetic Network 1, they exhibit a distinct split in the Haslemere network: the first four communities have closely aligned thresholds, whereas the remaining four vary widely and deviate further from $\tau$.

\subsection{Simulating the SIS Model on the Network}
In our second experiment, we simulate the discrete-time SIS dynamics \eqref{main-model} on two datasets \textbf{Network1} and \textbf{Haslemere}. 

Given the global infection probability vector $P(t)=(p_1(t),p_2(t),\ldots,p_n(t))^T.$ Let $p_g(t)=\left(p_i(t)\right)_{i\in C_g}$ denote the sub-vector restricted to community $C_g$. To quantify the aggregate infection quantity within community $C_g$ at time step $t$, we define the formula:
\[
Q_g(t)
=
\|p_g(t)\|_1
=
\sum_{i\in C_g}p_i(t).
\]

We simulate disease under two cases using an identical initial condition: $10$ nodes are randomly selected across the global network as initial seeds, each assigned an initial infection probability chosen randomly from $[0.87, 0.98]$.

\textbf{Case 1: }The effective spreading rate falls below the global epidemic threshold.

We fix the recovery rate at $\delta=0.4$ and choose the transmission rate $\beta$ such that $\frac{\beta}{\delta} = 0.95\,\tau.$

For \textbf{Network1}, we set $\frac{\beta}{\delta}=0.074$. For the \textbf{Haslemere} network, we set $\frac{\beta}{\delta}=0.15$.

\begin{figure}[H]
\centering

\begin{minipage}[t]{0.48\textwidth}
\centering
\includegraphics[width=\linewidth]{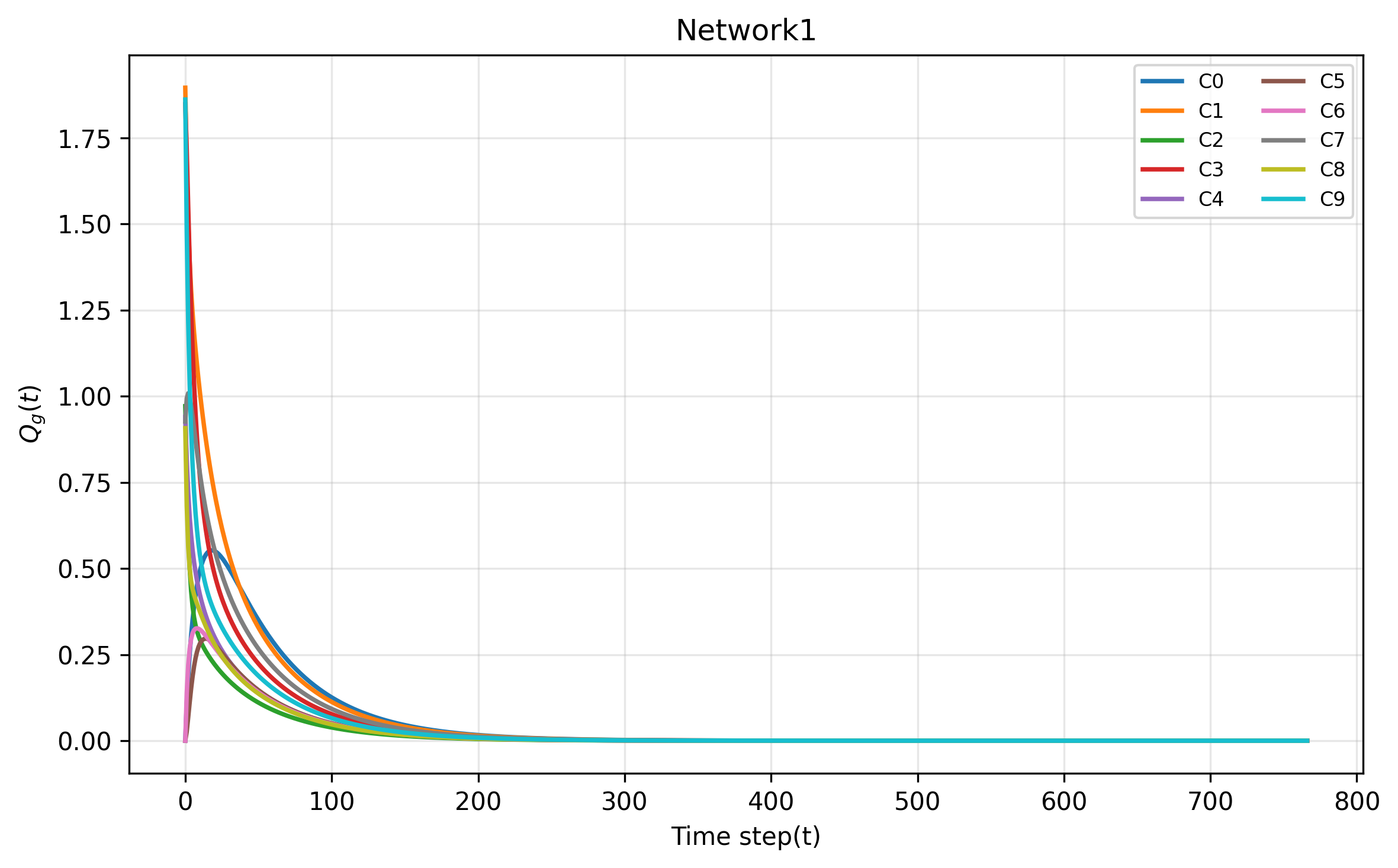}
\end{minipage}
\hfill
\begin{minipage}[t]{0.48\textwidth}
\centering
\includegraphics[width=\linewidth]{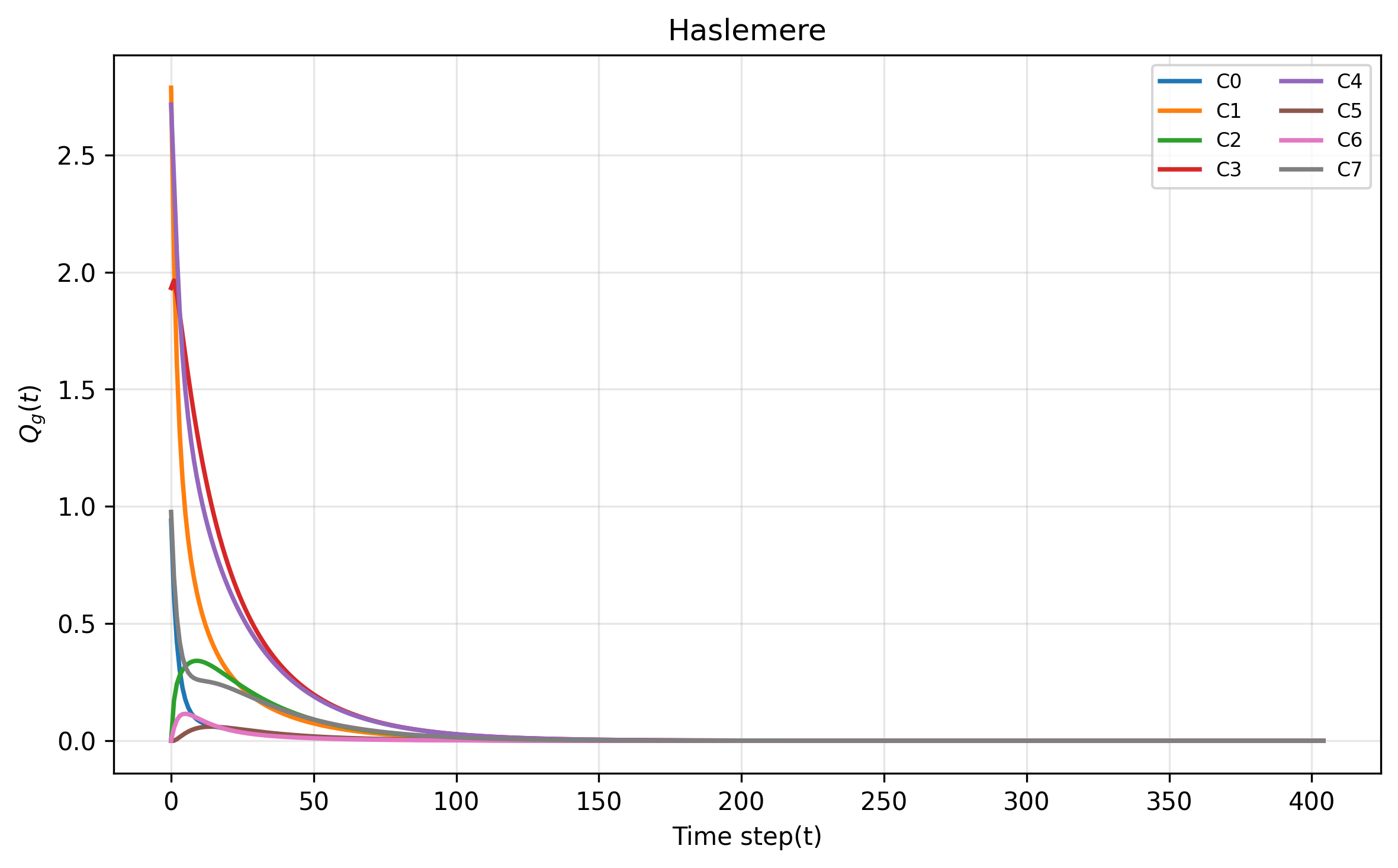}
\end{minipage}
\caption{SIS simulation and local community spreading trends for \textbf{Case 1}.}
\label{fig:Mô phỏng SIS}
\end{figure}

In this case, the disease fails to break out in any community, exactly as shown in Corollary \ref{hequa1}. The aggregate infection probabilities $Q_g(t)$ of all communities decay toward the disease-free equilibrium $0$.

\textbf{Case 2: }The effective spreading rate exceeds the local epidemic thresholds of several communities while falling below those of the remaining ones.

We fix the recovery rate at $\delta=0.4$ and choose $\beta$ such that $\frac{\beta}{\delta} = \frac{\tau_{\min}+\tau_{\max}}{2}$. Here, $\tau_{\min}$ and $\tau_{\max}$ denote the minimum and maximum local epidemic thresholds in the network, respectively.

For \textbf{Network1}, we set $\frac{\beta}{\delta}=0.12$. For the \textbf{Haslemere} network, we set $\frac{\beta}{\delta}=0.296$.

\begin{figure}[H]
\centering

\begin{minipage}[t]{0.48\textwidth}
\centering
\includegraphics[width=\linewidth]{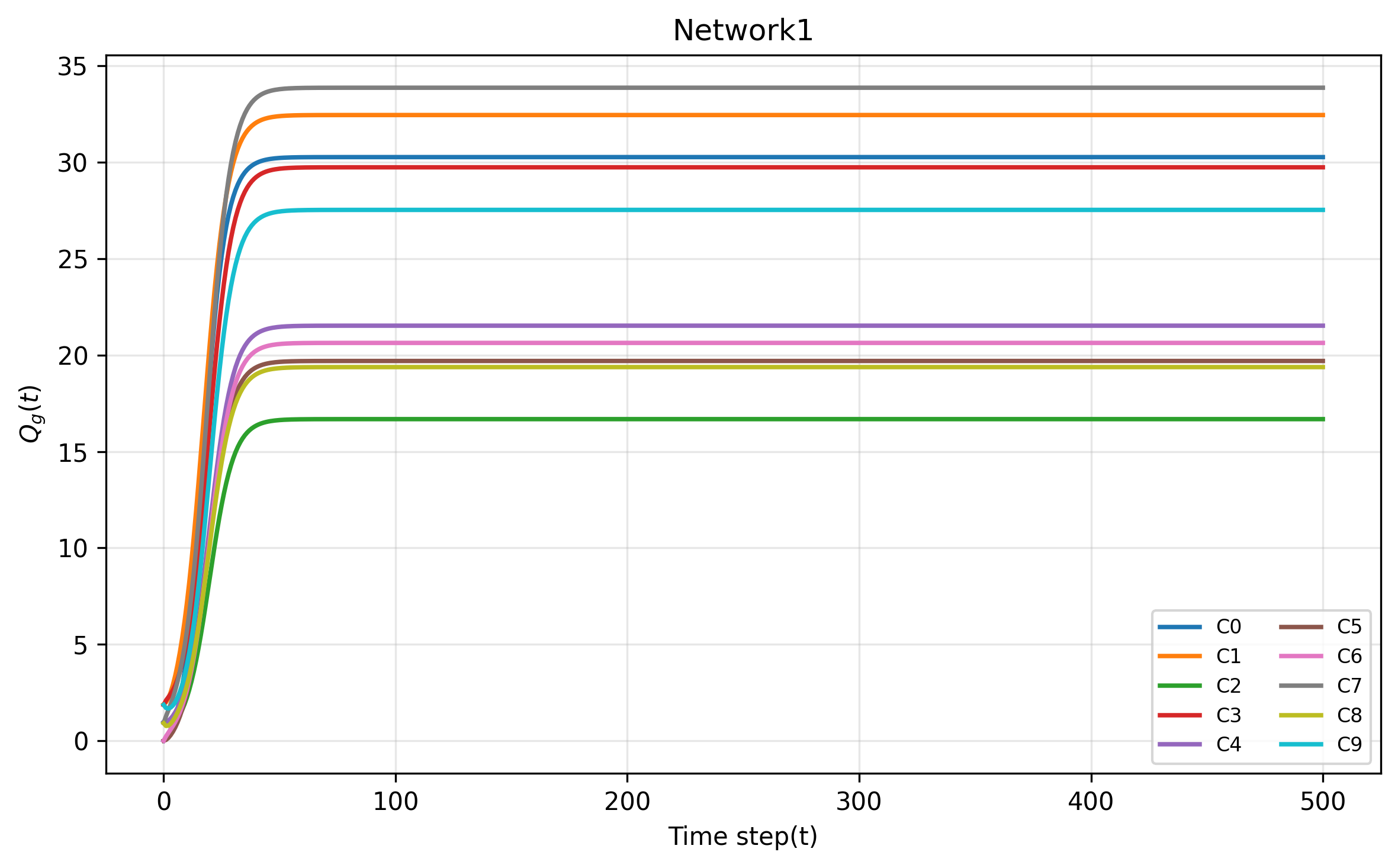}
\end{minipage}
\hfill
\begin{minipage}[t]{0.48\textwidth}
\centering
\includegraphics[width=\linewidth]{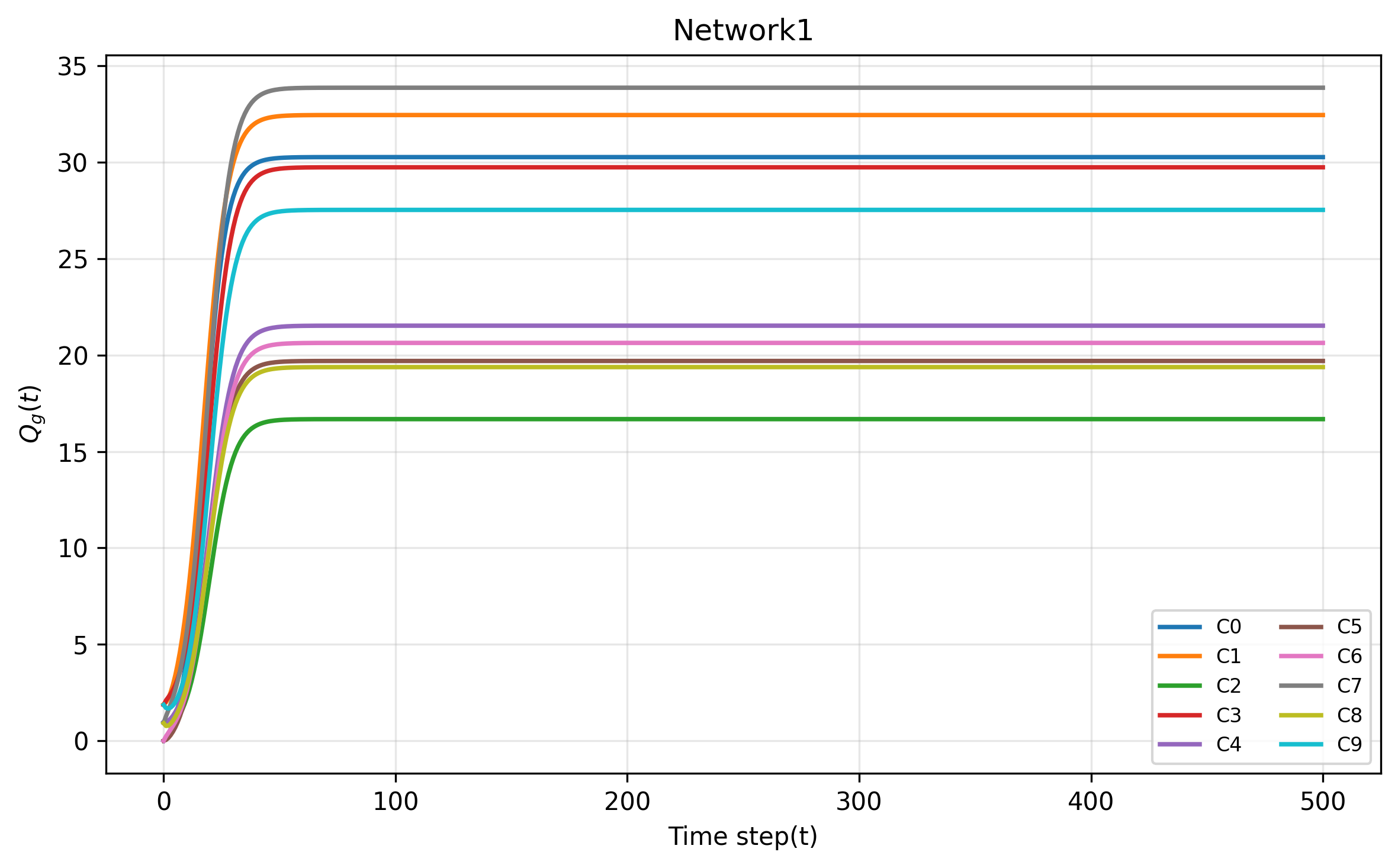}
\end{minipage}
\caption{SIS simulation and local community spreading trends for \textbf{Case 2}.}
\label{fig:Mô phỏng SIS trên ngưỡng}
\end{figure}

In this case, the selected $\beta$ and $\delta$ result in an effective spreading rate that exceeds the local epidemic thresholds of several communities while remaining below those of the remaining ones. Consistent with Corollary \ref{hequa2}, the epidemic breaks out across the entire network. The aggregate infection levels $Q_g(t)$ of all communities converge toward a non-zero equilibrium state.
\begin{remark}
Due to random seeding, initial infection levels $Q_g(0)$ vary, leaving some communities temporarily uninfected ($Q_g(0) = 0$). Within a few steps, however, all modules reach $Q_g(t) > 0$. This transition confirms that the infection inevitably spreads from seeded to unseeded communities via boundary edges.
\end{remark}

\section{Conclusion}\label{conclude}
In summary, this study established a community-level epidemic threshold set for discrete-time SIS dynamics, supported by numerical and simulation analyses. In future work, we plan to generalize this framework to alternative epidemic models and broader network topologies. Furthermore, this approach offers a scalable framework for investigating local propagation dynamics in massive networks by decomposing them into modular structures, which may inform future community-based risk assessment strategies.



\end{document}